\documentclass[journal,12pt,onecolumn,draftclsnofoot,]{IEEEtran}

\usepackage{graphicx}
\usepackage{amsmath}
\usepackage{amssymb }
\usepackage{graphicx}
\usepackage{amsmath}
\usepackage{amssymb}
\usepackage{amsthm}
\usepackage{lipsum}
\usepackage[dvipsnames]{xcolor}
\usepackage{subfig}

\newtheorem{thm}{Theorem}[section]

\theoremstyle{definition}

\usepackage[ruled,vlined]{algorithm2e}
\def\BibTeX{{\rm B\kern-.05em{\sc i\kern-.025em b}\kern-.08em
    T\kern-.1667em\lower.7ex\hbox{E}\kern-.125emX}}

\begin{document}

\title{Statistical CSI-based Beamforming for RIS-Aided
Multiuser MISO Systems using Deep
Reinforcement Learning}

\author{Mahdi Eskandari,
        Huiling Zhu,
        Arman Shojaeifard,
        and~Jiangzhou Wang,~\IEEEmembership{Fellow,~IEEE,}
        }


\maketitle

\begin{abstract}
The paper presents a joint beamforming algorithm using statistical channel state information (S-CSI) for reconfigurable intelligent surfaces (RIS) for multiuser MISO wireless communications. We used S-CSI, which is a long-term average of the cascaded channel as opposed to instantaneous CSI utilized in most existing works. Through this method, the overhead of channel estimation is dramatically reduced. We propose a proximal policy optimization (PPO) algorithm which is a well-known actor-critic based reinforcement learning (RL) algorithm to solve the optimization problem. To test the efficacy of this algorithm, simulation results are presented along with evaluations of key system parameters, including the Rician factor and RIS location, on the achievable sum rate of the users.
\end{abstract}


\IEEEpeerreviewmaketitle

\section{Introduction}
Reconfigurable intelligent surfaces (RIS) have recently emerged as an ideal candidate for delivering high data rate wireless services in next-generation wireless systems \cite{wu2019towards}.
An RIS
is a meta-surface comprising a large number of reflecting elements, capable of reflecting
the incident signal with a given phase shift. By densely deploying the RIS in wireless communication networks and intelligently coordinating their elements, the wireless
channels between the transmitter and receiver can be intentionally and deterministically
controlled to improve signal quality at the receiver and the network’s capacity \cite{huang2019reconfigurable, di2020smart}.
 Thus,  phase shift design  (passive beamforming) at the RIS, as well as active beamforming at the base station (BS), are of crucial importance in the design process. Several works have addressed the phase shift design problem under various configurations in response to this. For instance, in \cite{wu2019intelligent}, the optimization of active beamforming at the BS and passive beamforming at the RIS was explored where the generalization to discrete phase shifts proposed in \cite{wu2019beamforming}. In \cite{huang2019reconfigurable}, energy efficiency is maximized by jointly optimizing beamforming
at the IRS and the power allocation.
 In \cite{ wu2018intelligent, wu2019intelligent}, the authors design beamforming at the BS
and IRS with the goal of minimizing BS’ transmit power.
 In light of recent developments in machine learning, an unsupervised learning method has been proposed for passive beamforming design \cite{song2020unsupervised}. A joint beamforming and phase-shift design problem was also addressed with deep reinforcement learning in \cite{huang2020reconfigurable}.

  One of the biggest challenges at the BS is acquiring CSI (channel state information). Channel reciprocity has been exploited in time-division duplexing (TDD) systems, and uplink training can be used to obtain the instantaneous CSI at the BS \cite{marzetta2010noncooperative, rusek2012scaling}. With a non-orthogonal pilot in each cell in a heterogeneous and homogeneous network, a cell's estimation will be affected by pilot contamination from the other cells. \cite{jose2011pilot, wang2015acquisition}. 
  The results of \cite{marzetta2010noncooperative} showed that large-scale antenna systems can suffer from this pilot contamination effect when there are many antennas. In FDD systems, the BS can access the CSI through a feedback channel.
   A large number of antennas in the BS can make instant CSI feedback difficult on the feedback link. The BS has problems collecting accurate CSI information, especially when users are highly mobile. Another approach is to use the channel's second-order statistics. The statistical CSI varies significantly less rapidly than the instantaneous CSI \cite{adhikary2013joint}. In this way, the BS can obtain feedback over a longer period of time, reducing the amount of feedback.

Prior to this, the performance was optimised primarily using instantaneous CSI. Wireless communication based on RIS is frequently hindered by the difficulty of obtaining accurate instantaneous CSI \cite{wu2019towards}.
 This is because RIS, when it is in the reflecting mode, is a passive device. It is not capable of receiving or sampling incident signals, as opposed to conventional active antenna arrays.
 Using single-input, single-output (MIMO configuration only complicates things further), we will need to estimate $2M$ CSIs associated with $M$ IRS elements, plus one direct channel, and all through a single antenna, whether at receiver or transmitter.
 In order to avoid possible underdetermination in channel estimation, it may be necessary to transmit $2M+1$ pilot symbols with distinct associated IRS states \cite{dang2020joint}.
There have been several studies of instantaneous CSI \cite{you2020channel, wang2020channel, ning2020channel}, but it would be better to explore how to exploit statistical CSI, which is a much more stable quantity that changes very slowly over time.
 In the long run, it would be preferable to explore strategies for utilizing statistical CSI, which is a much more stable quantity that increases more slowly over time. In this case, channel statistics would be gathered and updated in sufficient time. The overhead of signaling exchange and computation burden can be reduced, as the state of RIS needs not to be changed frequently.Since the state of RIS does not need to be changed frequently, the overhead of signaling exchange and computation burden is reduced. Because the state of the RIS doesn't need to be changed frequently, the overhead of signaling exchange and computation burden can be reduced. Therefore, it is reasonable to utilize statistical CSI to determine RIS coefficients over an extended period of time.
To optimise the reflection coefficients of RIS-aided wireless systems, CSIs of the BS-RIS and RIS-user links must be accurate for the aforementioned reasons. Due to the high number of passive reflecting elements in the BS-RIS and RIS-UES links, the CSI may be difficult to determine in reality. As a result, it becomes imperative to maintain low training overhead while maximizing performance gains offered by RIS.
With the above works, CSI has been assumed to be readily available instantaneously to streamline the design process. Since the RIS is passive, it is extremely difficult to acquire accurate instantaneous CSI; thus, substantial training costs will be incurred. Therefore, a statistical CSI provides the most benefit in designing the system, because it varies much slower and can be obtained relatively easily \cite{gan2021ris}.  In works \cite{han2019large, hu2020statistical}, only S-CSI have been used to design passive beamformers for single-user multiple-input single-output (MISO) systems. For downlink transmission in multi-user scenarios, a two-timescale beamforming approach has been proposed \cite{zhao2021two}, in which RIS phase shifts are built up using statistical CSI, but BS transmit beamforming is built up using instantaneous CSI. 

Unlike most previous works, which rely on instantaneous CSI and S-CSI mixed types, we propose a pure S-CSI-based approach to designing multiuser RIS-assisted downlink systems. 
The main contributions of this paper are summarized as follows:
\begin{itemize}
    \item For multi-user MISO (MU-MISO) systems an approximate expression for the ergodic sum rate is developed. This leads to the formulation of a joint active-passive beam forming design problem.
    \item The objective function is intractable due to the non-convex constraint and the intricate relationship between the BS transmit beamformer and the RIS passive beamformer. To tackle this issue, a proximal policy optimization (PPO) algorithm has been proposed to solve the optimization problem. PPO is a powerful reinforcement learning algorithm from the actor-critic family, which has been demonstrated to outperform other actor-critic algorithms \cite{schulman2017proximal}.

\end{itemize}

The rest of this paper is organized as follows. Section II provides the system model, whereas the problem formulation is described in section III. Section IV gives the detailed PPO algorithm for solving the optimization problem. Section V provides simulation setup and results. Section VI concludes the paper.

\textit{Notations:} $x$ is a scalar, $\mathbf{x}$ denotes a vector, $\mathbf{X}$ is a matrix, and $\mathcal{A}$ represents a set. For a vector or matrix, the transpose and Hermitian (conjugate-transpose) operators are denoted by $(.)^T$ and $(.)^H$, respectively whereas the conjugate operator is given by $(.)^*$. $\mathbf{I}_N$ denotes an identity matrix
with subscript $N$ being the matrix dimension. Additionally, an all zero matrix of size $N \times N$ is represented by $\mathbf{0}_{N \times N}$. The trace operator of a matrix is denoted by $\mathrm{trace}(.)$. $\mathbb{C}$ represents the complex set and $\mathrm{diag}(\mathbf{a})$ denotes a diagonal matrix with entries $\mathbf{a}$ along its main diagonal. $\rVert \mathbf{x} \rVert$ denotes the Euclidean norm of $ \mathbf{x}$ and also $|x|$ shows the absolute value of complex scalar $x$. $\mathbb{E}[.]$ denotes expectation

\section{System Model}
\begin{figure}[t]
    \centering
    \includegraphics[width=0.7\textwidth]{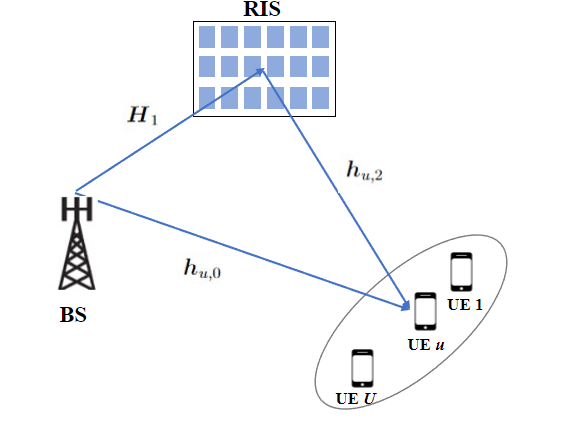}
    \caption{Illustration of RIS-aided multi-user MISO communications system}
    \label{fig:sys}
\end{figure}

As shown in Fig.~\ref{fig:sys}, a multi-user MISO (MU-MISO) communication system with $N$ antennas at the transmitter which is communicating with $U$ single-antenna user equipments (UE) is considered. Transmission is done with the assistance of a reconfigurable intelligent surface (RIS) which is equipped with $M$ reflecting elements. By denoting $\xi_m$ as the reflection coefficient of the $m$-th element of RIS, the reflection matrix of the RIS panel can be expressed as $\boldsymbol{\Xi} = \mathrm{diag}(\xi_1, \dots, \xi_M)$ where $\xi_m = e^{j \phi_m}$ with $\phi_m$ denoting the phase shift of $m$-th element of RIS. Denote $\mathbf{h}_{u,0} \in \mathbb{C}^{N \times 1}$ as the direct channel from the transmitter to the $u$-th UE, $\mathbf{H}_1 \in \mathbb{C}^{M \times N}$ that from transmitter to RIS and finally $\mathbf{h}_{u,2} \in \mathbb{C}^{M \times 1}$ that from RIS to the $u$-th UE. Hence, the equivalent effective channel from the transmitter to the $u$-th UE would be $\mathbf{h}_u^T \triangleq  \mathbf{h}_{u,0}^T + \mathbf{h}_{u,2}^T \boldsymbol{\Xi} \mathbf{H}_1$.

We will now take a closer look at channel models. The Rician distribution is used to model the channels. The channel matrix $\mathbf{h}_{u,2}$ between the RIS and $u$-th UE is represented as \cite{gan2021ris, wang2021joint}
\begin{equation}
    \mathbf{h}_{u,2} = \sqrt{\frac{\delta_{u,2} \kappa_{u, 2}}{1+\kappa_{u, 2}}} \mathbf{\bar{h}}_{u,2} + \sqrt{\frac{\delta_{u,2}}{1+\kappa_{u, 2}}} \mathbf{\tilde{h}}_{u,2},
\end{equation}
where $\sqrt{\delta_{u,2}}$ denotes the distance dependent path-loss factor, also, $\kappa_{u, 2}$ denotes the Rician factor between the RIS and the $u$-th UE. Furthermore, $\Bar{\mathbf{h}}_{u,2} = \mathbf{a}_\text{RIS}(\phi_u^{(\text{RIS})}, \psi_u^{(\text{RIS})})$ and $\phi_u^{(\text{RIS})}$ ($\psi_u^{(\text{RIS})}$) is the azimuth (elevation) angle of departure (AoD) from the RIS to the $u$-th UE. $\mathbf{a}_\text{RIS}(\phi_u^{(\text{RIS})}, \psi_u^{(\text{RIS})})$ is the array response vector at the RIS side. By assuming a  uniform planar arrays (UPA), $\mathbf{a}_\text{RIS}(\phi_u^{(\text{RIS})}, \psi_u^{(\text{RIS})})$ is given by
 \begin{align}
     \mathbf{a}_\text{RIS}(\phi_u^{(\text{RIS})}, \psi_u^{(\text{RIS})}) = \frac{1}{\sqrt{M_{\text{H}} M_\text{V}}} [1,\dots,& e^{j \frac{2 \pi}{\lambda}D(h \sin{\phi_u^{(\text{RIS})} \sin (\psi_u^{(\text{RIS})})}+v \cos{\psi_u^{(\text{RIS})}})} ,\dots, \\ \nonumber
     &e^{j \frac{2 \pi}{\lambda}D((M_{\text{H}}) \sin{\phi_u^{(\text{RIS})} \sin (\psi_u^{(\text{RIS})})}+ (M_{\text{V}}-1) \cos{\psi_u^{(\text{RIS})}})}]^T,
 \end{align}
  where $\lambda$ is the wavelength, $D$ is the distance between antenna elements and also $0 \leq h \leq M_{\text{H}}-1$, $0 \leq v \leq M_{\text{V}}-1$ with $M = M_{\text{H}} M_{\text{V}}$ with $M_{\text{H}}$ and $M_{\text{V}}$ being the number of elements at horizontal and vertical axis, respectively. The array response vector of the BS side can be written in a similar fashion with appropriate change in notations. 
 
 On the other side, the channels between the transmitter and RIS, and direct one from transmitter to the $u$-th receiver could be expressed as \cite{gan2021ris, wang2021joint, zhang2020capacity, zhou2020joint}  
 \begin{align}
     \mathbf{H}_1 &= \sqrt{\frac{ \delta_1 \kappa_1}{1+\kappa_1}}  \Bar{\mathbf{H}}_1 + \sqrt{\frac{\delta_1}{1+\kappa_1}}  \Tilde{\mathbf{H}}_1, \\
    \mathbf{h}_{u,0} &= \sqrt{\frac{ \delta_{u,0} \kappa_{u,0}}{1+\kappa_{u,0}}}  \Bar{\mathbf{h}}_{u,0} + \sqrt{\frac{\delta_{u,0}}{1+\kappa_{u,0}}}  \Tilde{\mathbf{h}}_{u,0},  
 \end{align}
  with $\delta_1$ and $\kappa_1$ being the distance dependent pathloss and the Rician factor of the BS-RIS link and, moreover, $\delta_{u, 0}$ and $\kappa_{u, 0}$ denote the distance dependent pathloss and Rician factor between BS and the $u$-th UE link.
 Furthermore, $\Bar{\mathbf{H}}_1 = \mathbf{a}_\text{RIS}(\phi^{(\text{RIS})}, \psi^{(\text{RIS})}) \mathbf{a}_\text{BS}(\phi^{(\text{BS})}, \psi^{(\text{BS})})^H$ where $\phi^{(\text{RIS})}$ and $\psi^{(\text{RIS})}$ being the azimuth and elevation angle of angle of arrival (AoA) to the RIS and $\phi^{(\text{BS})}$ and $\psi^{(\text{BS})}$ denote the azimuth and elevation AoD from the BS to the RIS direction. On the other hand, $ \Bar{\mathbf{h}}_{u,0} = \mathbf{a}_\text{BS}(\phi_u^{(\text{BS})}, \psi_u^{(\text{BS})})$ where $\phi_u^{(\text{BS})}$ and $\psi_u^{(\text{BS})}$ being the azimuth and elevation AoD from the BS in the direction of the $u$-th UE. 
 
 By considering the fact that there is no spatial correlation among the antennas, the distribution of non-line-of-sight (NLoS) component of the channels, $\mathbf{\tilde{h}}_{u,2}$, $ \Tilde{\mathbf{h}}_{u,0}$ and $ \Tilde{\mathbf{H}}_1$ would be independently and identically distributed (i.i.d.)
complex Gaussian random variables with the zero mean and
unit variance.
 
 Based on the above descriptions, the received signal for the $u$-th UE $y_u \in \mathbb{C}$ is given by
 \begin{equation}
     y_u = \sum_{u=1}^{U} \sqrt{P_u} \mathbf{h}_u^T \mathbf{f}_u x_u + n_u,
 \end{equation}
 where $P_u$ and $x_u$ are the allocated power and signal for $u$-th UE, respectively, furthermore, $\mathbf{f}_u$ is the beamforming vector for the $u$-th UE and $n_u \in \mathbb{C}$ is a circularly symmetric complex additive Gaussian noise with $\mathbb{E}[n_u n_u^H] = \sigma^2$.
 
  \section{Problem Formulation}
  
 In this paper it is assumed that the transmitter doesn't complete information about CSI, which is really time consuming for estimation especially in RIS-aided environments. Instead, the sum~rate maximization problem is formulated based on statistical components of the channel, namely angular information and Rician factors of the channels which are easier to estimate through feedback link \cite{liu2018spectral}.

 The ergodic rate in $ \mathrm{[bit/sec/Hz]}$ of random MIMO channel with equal power allocation for the $u$-th UE is given by
 \begin{equation}
     R_u = \mathbb{E}\left[\log_2 \left(1 + \frac{\frac{P}{U} |\mathbf{h}_u^T \mathbf{f}_u|^2}{\sigma_2 + \sum_{i \neq u} \frac{P}{U} |\mathbf{h}_u^T \mathbf{f}_i|^2} \right)  \right] = \mathbb{E}\left[\log_2 \left(1 + \frac{\frac{P}{U} \mathbf{f}_u^H \mathbf{h}_u^* \mathbf{h}_u^T \mathbf{f}_u}{\sigma_2 + \frac{P}{U} \sum_{i \neq u} \mathbf{f}_i^H \mathbf{h}_u^* \mathbf{h}_u^T \mathbf{f}_i} \right)  \right],
     \label{capacity}
 \end{equation}
where the expectation is over all the channel realisations and $(.)^*$ denotes the conjugate operator. The ergodic rate of the system is represented by
 \begin{equation}
R = \sum_{u = 1}^{U} R_u.
 \end{equation}
 Finally, the ergodic sum-rate maximisation problem could be formulated as
  \begin{subequations} \label{eq:litdiff}
 \begin{align}
   \mathcal{P}_1: \hspace{5mm} \max_{\boldsymbol{\Theta}, \mathbf{f}_u} \hspace{5mm} &\sum_{u=1}^{U}  {\mathbb{E}\left[\log_2 \left(1 + \frac{\frac{P}{U} \mathbf{f}_u^H \mathbf{h}_u^* \mathbf{h}_u^T \mathbf{f}_u}{\sigma_2 + \sum_{i \neq u} \frac{P}{U} \mathbf{f}_i^H \mathbf{h}_u^* \mathbf{h}_u^T \mathbf{f}_i} \right)  \right]} \tag{\ref{eq:litdiff}} \\
 \text{subject to} \hspace{10mm} &|\xi_m|^2 = 1, \hspace{3mm}m = 1, \dots, M \\
 &\Vert \mathbf{f}_u \Vert_2^2 = 1, \hspace{3mm}u = 1, \dots, U
 \end{align}
 \end{subequations}
 The maximisation problem $\mathcal{P}_1$ is mathematically intractable duo to existence of the expectation and unit norm constraints. In what follows, we will derive the expectation and find the upper bound of it. 
 
To begin with, (\ref{capacity}) can be re-written as

\begin{equation}
    R_u = \mathbb{E}\left[\log_2 \left({ \sigma^2 + \frac{P}{U} \sum_{i=1}^{U} \mathbf{f}_u^H \mathbf{h}_u^* \mathbf{h}_u^T \mathbf{f}_u} \right)  \right] - \mathbb{E}\left[\log_2 \left({\sigma^2 + \frac{P}{U} \sum_{i\neq u} \mathbf{f}_i^H \mathbf{h}_u^* \mathbf{h}_u^T \mathbf{f}_i} \right)  \right]
\end{equation}
Next, by using Jensen's inequality, we obtain
\begin{equation}
    R_u  \leq \log_2 \left({ \sigma^2 + \frac{P}{U} \sum_{i=1}^{U} \mathbf{f}_u^H \mathbb{E} \left[\mathbf{h}_u^* \mathbf{h}_u^T \right] \mathbf{f}_u} \right) - \log_2 \left({\sigma^2 + \frac{P}{U} \sum_{i\neq u} \mathbf{f}_i^H \mathbb{E}\left[ \mathbf{h}_u^* \mathbf{h}_u^T \right] \mathbf{f}_i} \right)  
    \label{exp_in}
\end{equation}

 Based on (\ref{exp_in}), the problem turns into finding $\mathbb{E}[ \mathbf{h}_u^* \mathbf{h}_u^T ]$. Next, we focus on derivation of $\mathbb{E}[ \mathbf{f}_u^H \mathbf{h}_u^* \mathbf{h}_u^T \mathbf{f}_u ]$ where the later one has the similar approach. The equivalent channel vector $\mathbf{h}_u$ can be re-written as
 
\begin{thm} \label{th1}
$\mathbb{E}[ \mathbf{h}_u^* \mathbf{h}_u^T ]$ could be approximated as follows
\begin{align}
    \mathbb{E}[ \mathbf{h}_u^* \mathbf{h}_u^T ] \triangleq \mathbf{C}_u &= \frac{\delta_{u,0} \kappa_{u,0}}{1 +  \kappa_{u,0}} \Bar{\mathbf{h}}_{u,0}^* \Bar{\mathbf{h}}_{u,0}^T + \left( \frac{\delta_{u,0}}{1 +  \kappa_{u,0}} + \frac{M\delta_{1}\delta_{u,2}}{1 + \kappa_1} \right) \mathbf{I}_N \\ \nonumber
    &+ \sqrt{\frac{\delta_{1} \delta_{u,0} \delta_{u,2} \kappa_{1} \kappa_{u,0} \kappa_{u,2}}{(1 + \kappa_{1})(1 + \kappa_{u,0})(1 + \kappa_{u,2})}} \left(  \Bar{\mathbf{h}}_{u,0}^* \Bar{\mathbf{h}}_{u,2}^T \boldsymbol{\Xi} \Bar{\mathbf{H}}_1 +  \Bar{\mathbf{H}}_1^H \boldsymbol{\Xi}^* \Bar{\mathbf{h}}_{u,2}^* \Bar{\mathbf{h}}_{u,0}^T \right) \\ \nonumber
    &+ \frac{M\delta_{u,2} \delta_{1} \kappa_{1}}{(1 + \kappa_{1})(1 + \kappa_{u, 2})}\mathbf{a}_\text{BS}(\phi^{(\text{BS})}, \psi^{(\text{BS})})^H \mathbf{a}_\text{BS}(\phi^{(\text{BS})}, \psi^{(\text{BS})}) \\ \nonumber
    & + \frac{\delta_{1} \delta_{u,2}  \kappa_{1} \kappa_{u,2}}{(1 +  \kappa_{1})(1 +  \kappa_{u, 2})} \Bar{\mathbf{H}}_1^H \boldsymbol{\Xi}^* \Bar{\mathbf{h}}_{u,2}^* \Bar{\mathbf{h}}_{u,2}^T \boldsymbol{\Xi}  \Bar{\mathbf{H}}_1
\end{align}
\end{thm}
\begin{proof}
Refer to Appendix 1, please. 
\end{proof}
 
According to theorem \ref{th1}, $\mathbf{C}_u$ only depends on the statistics of the channel, i.e. the angles of departure and arrival, Rician factors, and also the sum rate is dependent upon the active beamforming vector at the BS and the passive beamforming matrix at the RIS, therefore both active and passive beamformers should be designed carefully. 

 \begin{subequations} \label{eq:opt2}
 \begin{align}
  \mathcal{P}_2: \hspace{5mm} \max_{\boldsymbol{\Theta}, \mathbf{f}_u} \hspace{5mm} &\sum_{u=1}^{U}  \log_2 \left(1 + \frac{\frac{P}{U} \mathbf{f}_u^H  \mathbf{C}_u \mathbf{f}_u}{\sigma_2 + \sum_{i \neq u} \frac{P}{U} \mathbf{f}_i^H \mathbf{C}_u \mathbf{f}_i} \right) \tag{\ref{eq:opt2}} \\
 \text{subject to} \hspace{10mm} &|\xi_m|^2 = 1, \hspace{3mm}m = 1, \dots, M \\
 &\Vert \mathbf{f}_u \Vert_2^2 = 1, \hspace{3mm}u = 1, \dots, U
 \end{align}
 \end{subequations}
 As it is obvious from problem $\mathcal{P}_2$, the new optimisation problem is now a function of distance dependent path-loss, Rician factors and other statistical components of the channel and the effect of small scale fading has been averaged out. 
 
 In the next section, we will propose a PPO algorithm to solve problem $\mathcal{P}_2$.
 
 \section{Proximal Policy Optimization Approach}
 \subsection{PPO Background}
 In this section, PPO algorithm is used to solve optimisation problem $\mathcal{P}_2$. PPO was first introduced in \cite{schulman2017proximal} and, moreover, this algorithm has been shown to perform better than the other algorithms on the benchmark while being easy to tune, efficient, and simple to implement. PPO is a model-free, on-policy, actor-critic, policy gradient method. 
 
 Consider an infinite-horizon discounted Markov decision process (MDP), defined by the tuple $(\mathcal{S}, \mathcal{A}, P, r, \gamma)$, where $\mathcal{S}$ being the finite set of states, $\mathcal{A}$ is the finite set of actions, $P: \mathcal{S} \times \mathcal{A} \times \mathcal{S} \rightarrow \mathbb{R}$
 is the transition probability distribution, $r: \mathcal{S} \rightarrow \mathbb{R}$ is the reward function and finally, $\gamma \in [0, 1]$ is the discount factor.
PPO purpose is to retain the reliability of trust region policy optimization (TRPO) algorithms, which guarantee monotonic improvements when taking into account the Kullback-Leibler (KL) divergence of policy updates while only using first-order optimisation techniques. TRPO maximizes the objective function while taking account of a constraint on the size of the policy update. Specifically,
\begin{subequations} \label{eq:opt}
\begin{align}
    \max_{\boldsymbol{\theta}}& \hspace{2mm} \hat{\mathbb{E}}_t \left[ \frac{\pi_{\boldsymbol{\theta}}(a_t | s_t)}{\pi_{\boldsymbol{\theta}_\text{old}}(a_t | s_t)} \hat{A}_t  \right]  \tag{\ref{eq:opt}} \\
    \text{subject to} &\hspace{10mm} \hat{\mathbb{E}} \left[ \mathrm{KL} \left[   \pi_{\boldsymbol{\theta}_\text{old}}(. | s_t), \pi_{\boldsymbol{\theta}}(. | s_t) \right]   \right] \leq \delta,
\end{align}
 \end{subequations}
where the expectation $\hat{\mathbb{E}}_t[.]$ indicates the empirical average over a finite batch of samples, in an algorithm that alternates between sampling and optimization. $\pi_{\boldsymbol{\theta}}$ is the policy neural network with parameters denoted by $\boldsymbol{\theta}$ and $\pi: \mathcal{S} \times \mathcal{A} \rightarrow [0, 1]$, $a_t$ and $s_t$ are the action and state of given timestep $t$, respectively with $a_t \sim \pi_{\boldsymbol\theta}(a_t | s_t) $, $\boldsymbol{\theta}_\text{old}$ is the vector of policy parameters before the update. Furthermore, $\hat{A}_t$ is an estimator of the advantage function at timestep $t$ and is given by
\begin{equation}
    A_t = Q(s_t, a_t) - V(s_t),
\end{equation}
where $Q(., .)$ and $V(.)$ are the action-value the value functions, respectively and are defined as follows
\begin{align}
     Q(s_t, a_t) &= \mathbb{E}_{s_{t+1}, a_{t+1}, \dots} \left[ \sum_{\ell = 0}^{ \infty} \gamma^\ell r(s_{t+1}) \right], \\
     V(s_t) &= \mathbb{E}_{a_t, s_{t+1}, \dots}  \left[ \sum_{\ell = 0}^{ \infty} \gamma^\ell r(s_{t+1}) \right],
\end{align}
where $s_{t+1} \sim P(s_{t+1} | s_t, a_t) $. The estimate of the advantage function in the interval $t \in [0, T]$ is given by \cite{schulman2017proximal}
\begin{equation}
    \hat{A}_t = \delta_t + (\gamma \lambda)\delta_{t+1} + \dots + (\gamma\lambda)^{T-t+1}\delta_{T-1}
    \label{a_hat}
\end{equation}
with $\delta_t = r_t + \gamma V(s_{t+1}) - V(s_t) $ and $\lambda$ being a hyperparameter and denotes the factor for trade-off of bias and variance for generalized advantage estimator (GAE).

Next, let $\rho_t(\boldsymbol{\theta})$ denote the probability ratio $\rho_t(\boldsymbol{\theta}) =  \frac{\pi_{\boldsymbol{\theta}}(a_t | s_t)}{\pi_{\boldsymbol{\theta}_\text{old}}(a_t | s_t)}$, then obviously $\rho_t(\boldsymbol{\theta}_{\mathrm{old}}) = 1$, hence, according to (\ref{eq:opt}), TRPO maximises 
\begin{equation}
    \mathcal{L}^{\mathrm{CPI}}(\boldsymbol{\theta}) = \hat{\mathbb{E}}_t \left[ \rho_t(\boldsymbol{\theta})\hat{A}_t  \right] 
    \label{CPI}
\end{equation}
where CPI refers to conservative policy iteration. The main issue of the optimisation problem in (\ref{CPI}) is the probability of large policy update. For instance, if $\pi_{\boldsymbol{\theta}_\text{old}}(a_t | s_t)$ has a small value where as $\pi_{\boldsymbol{\theta}}(a_t | s_t)$ has a relatively large value, the value of $\rho_t(\boldsymbol{\theta})$ will tend to be really large and lead to taking big gradient steps that might cause the policy change in drastic ways. 
To solve this issue, PPO modifies the objective function in (\ref{CPI}) as follows
\begin{equation}
    \mathcal{L}^{\mathrm{CLIP}}(\boldsymbol{\theta}) = \hat{\mathbb{E}}_t \left[ \min \left( \rho_t(\boldsymbol{\theta})\hat{A}_t, \mathrm{clip}\left(\rho_t(\boldsymbol{\theta} \right), 1-\epsilon_c, 1+\epsilon_c) \hat{A}_t \right)  \right] 
    \label{CLIP}
\end{equation}
where $\epsilon_c$ is a hyperparameter. This objective is motivated by the following reasons. The first term in the $\min$ operator is the $\mathcal{L}^{\mathrm{CPI}}$ and the second term guarantees the probability ratio to be inside the interval $[1-\epsilon_c, 1+\epsilon_c]$ with the help of $\mathrm{clip}(., ., .)$ function. The $\mathrm{clip}(., ., .)$ function saturates the variable in the first input between the values of the second and third input. Hence, the final objective function is the minimum of clipped and unclipped objective and as a result provides a lower bound on the unclipped objective which leads to the actions with a negative advantage function are eliminated as a preference. With a small set of sample sizes, the algorithm isn't too greedy in favoring actions with positive advantage functions, nor too quick in avoiding actions with negative advantage functions. 
\subsection{Joint active and passive beamforming using PPO}
Generally, PPO is presented as an MDP with observation and action spaces. When solving the joint active and passive beamforming problem, the BS, RIS, and all the users in the system are denoted by the environment $\mathcal{E}$, while agent is BS which is able to control the RIS. The following are the key PPO elements that are employed to solve the joint active and passive beamforming problem.
\subsubsection{Observation space} At each timestep $t$, the observation part consist of four parts, first, it contains the real and imaginary parts of the beamforming vector $\mathbf{f}_u$, i.e., $\mathcal{F} = \{\mathcal{F}_\mathrm{r}, \mathcal{F}_\mathrm{i} \} $ with $\mathcal{F}_\mathrm{r} = \{\Re(\mathbf{f}_1), \dots, \Re(\mathbf{f}_U) \}$ and $\mathcal{F}_\mathrm{i} = \{\Im(\mathbf{f}_1), \dots, \Im(\mathbf{f}_U) \}$. The second part is the real and imaginary parts of the phase shifts of the RIS, i.e., $\mathcal{R} = \{ \Re(\mathrm{diag}(\boldsymbol{\Theta})), \Im(\mathrm{diag}(\boldsymbol{\Theta})) \}$. The third part of the observation is $\mathcal{M} = \{\boldsymbol{\mu}_1, \dots, \boldsymbol{\mu}_U \}$. Next, the forth part of the observation is $\mathcal{N} = \{\Lambda_1, \dots, \Lambda_U \}$. Finally the observation vector at timestep $t$ is as follows
\begin{equation}
    s_t = \{\mathcal{F}, \mathcal{R}, \mathcal{M}, \mathcal{N} \},
\end{equation}
Hence, the observation shape is $U(3N+1)+2M$.

\subsubsection{Action space} At each timestep $t$ the action space is the vector containing the real and imaginary parts of the beamforming vectors for all the users and the real and imaginary parts of the phase shifts of the RIS. Thus, the action shape is $2UN + 2M$ and the action range is $[0, 2\pi]$. Since $\mathrm{tanh}$ is the activation function for the final layer, which produces the values between $-1$ and $+1$, for converting the result to the desired action range, it is sufficient to set $a_t = \pi(a_t^{\prime}+1)$ where $a_t^{\prime}$ is the output of $\mathrm{tanh}$ activation layer. After building the real and imaginary parts of active and passive beamforming vectors, for guaranteeing the unit module constraint  it is sufficient to set $\mathbf{f}_u = \cos(\mathbf{f}_u^\mathrm{r}) + j \sin(\mathbf{f}_u^\mathrm{i})$, where $\mathbf{f}_u^\mathrm{r}$ and $\mathbf{f}_u^\mathrm{i}$ are the real and imaginary part of the output action for the $u$-th UE. Similarly, for the RIS, $\boldsymbol{\Theta} = \mathrm{diag}(\cos(\boldsymbol{\Theta}^\mathrm{r}) + j\sin(\boldsymbol{\Theta}^\mathrm{i}))$, where $\boldsymbol{\Theta}^\mathrm{r}$ and $\boldsymbol{\Theta}^\mathrm{i}$ are the real and imaginary part of the built from output action corresponding to the RIS part. 
\subsubsection{Reward function} At each timestep $t$ the reward function is the sum rate rate of all the users, specifically,
\begin{equation}
    r_t = \sum_{u=1}^U R_u(t)
\end{equation}
where $R_u(t)$ is the rate of $u$-th user at time step $t$.

The details of the proposed PPO algorithm for joint active ans passive beamforming are presented in Algorithm~\ref{alg_1}.

\begin{algorithm}[h]
\caption{PPO, Actor-Critic Style}
\label{alg_1}
\SetAlgoLined
 \textbf{Initialisation:} Initialise  time, states, actions, and replay buffer $\mathcal{D}$ for storing the random states, action and reward in each time step\;
 Randomly initialise the actor network and the critic network \;
 \textbf{Output:} The beamforming vectors of all the users and the phase shifts of the RIS \;
 \For{episode $j = 1, \dots, J$}{
  Initialise the environment $\mathcal{E}$ and make the initial state $s_0$\;
   Run the policy $\pi_{\boldsymbol{\theta}_\mathrm{old}}$ for $T$ timesteps \;
   Compute advantage estimates $\hat{A}_1, \dots, \hat{A}_T$ using (\ref{a_hat}) \;
   Optimise surrogate  $\mathcal{L}^{\mathrm{CLIP}}(\boldsymbol{\theta})$ w.r.t  $\boldsymbol{\theta}$, with $K$ epochs and minibatch size $M \leq T$ using (\ref{CLIP})\;
    $\boldsymbol{\theta}_{\mathrm{old}} \leftarrow \boldsymbol{\theta} $
 }
$\text{ Save the PPO model}$ \;
\end{algorithm}

\section{Simulation Results}
In this section, simulation results are provided to evaluate the performance of our proposed algorithm. The simulation environment is run on Python 3.10.0 with PyTorch v1.11.0 on a computer with AMD Ryzen 7PRO Eight-Core Processor 3.20 GHz CPU and $32$ GB of memory. In the simulation, the BS is located at $[5, 0, 30]$ at the Cartesian coordination whereas the RIS is at $(0, 70, 3)$, finally the users are randomly located at a circle with raduis of $3 \mathrm{m}$ centreing at $[5, 70, 0]$ and the number of users in the simulation is set to be $3$. The large-scale path loss is given by $PL [\mathrm{dB}] = PL_0 - 10\alpha \log_{10} (d/d_0)$, where $PL_0$ is the path loss at the reference distance $d_0=1$ which is given by $PL_0 = -30 \mathrm{dB}$, also $\alpha$ is the path loss exponent, and
$d$ is the distance between transmitter and receiver pair. The path
loss exponents of the BS-UEs, BS-RIS and RIS-UEs links are doneted by $\alpha_0$, $\alpha_1$, and $\alpha_2$, respectively. We set $\alpha_0 = 3.4$, $\alpha_1 = 2.2$ and $\alpha_1 = 3$, which means the path loss exponent of the BS-users link is larger than the other links. The number of BS antennas at the horizontal and vertical axis is $8$ and $4$, respectively, thus the total number of BS antennas are $N = 32$. The RIS consists of $8$ horizontal and $4$ vertical elements, which mean the total number of elements are $M = 32$. The noise power density is set to be $-80 \mathrm{dBm}$, the maximum transmit power at the BS is $10\mathrm{dBm}$.
The Rician factor between BS and UEs are set to be $\kappa_{u, 0} = -3 \mathrm{dB}, u = 1, \dots, U$ and that of between BS and RIS and RIS and UEs is $\kappa_{1} = \kappa_{u, 2} = 10 \mathrm{dB}, u = 1, \dots, U$. A typical realisation of the environment is illustrated in Fig.~\ref{env}. At the beginning of each episode a new environment similar to Fig.~\ref{env} with the fixed locations of the BS and RIS and randomly changing UEs positions is generated and the learning process goes on with the new generated environment. 
The hyperparameters for training the agents for all of the scenarios are listed in Table \ref{table:rl}
\begin{table*}[th]
\caption{Common PPO Hyperparameters} 
\centering 
\begin{tabular}{l r l r} 
\hline\hline 
Parameter & Value & Parameter & Value \\ [0.5ex] 
\hline 
Number of first fully connected layers & $(64, 64)$ & Discount factor $\gamma$ & $0.995$ \\
Trade-off factor GAE $\lambda$ & $0.95$ & Clip range $\epsilon_c$ & $0.2$  \\
Maximum time-steps of each episode & $4000$ & Number of episodes & $2000$ \\
Activation function for hidden layers & ReLU & Activation function for output layer & tanh \\
Horizon ($T$) & $10$ & Adam stepsize (Learning rate) & $0.00015$ \\
Number of epochs & $100$ &  Minibatch size & $10$ \\
 [1ex]
\hline 
\end{tabular}
\label{table:rl} 
\end{table*}

In order to evaluate the performance of the proposed algorithms, we used the following two baselines:
\begin{itemize}
    \item \textbf{Baseline 1} (Without RIS): In this case, The RIS is removed from the system setup by letting $M = 0$ and the PPO is used to solve the new optimisation problem $\mathcal{P}_3$. In this case $\mathbf{C}_u$ reduced to the following
    \begin{equation}
        \mathbf{C}_u = \frac{\delta_{u,0} \kappa_{u,0}}{1 +  \kappa_{u,0}} \Bar{\mathbf{h}}_{u,0}^* \Bar{\mathbf{h}}_{u,0}^T + \left( \frac{\delta_{u,0}}{1 +  \kappa_{u,0}}\right) \mathbf{I}_N,
    \end{equation}

and thus, the optimisation problem problem $\mathcal{P}_3$ is formulated as follows
 \begin{subequations} \label{eq:opt3}
 \begin{align}
  \mathcal{P}_3: \hspace{5mm} \max_{\boldsymbol{\Theta}, \mathbf{f}_u} \hspace{5mm} &\sum_{u=1}^{U}  \log_2 \left(1 + \frac{\frac{P}{U} \mathbf{f}_u^H  \mathbf{C}_u \mathbf{f}_u}{\sigma_2 + \sum_{i \neq u} \frac{P}{U} \mathbf{f}_i^H \mathbf{C}_u \mathbf{f}_i} \right) \tag{\ref{eq:opt3}} \\
 \text{subject to} \hspace{10mm}  &\Vert \mathbf{f}_u \Vert_2^2 = 1, \hspace{3mm}u = 1, \dots, U,
 \end{align}
 \end{subequations}
finally, the PPO is used to solve the problem $\mathcal{P}_3$.

    \item \textbf{Baseline 2} (Random Phase): In this scenario, the phase shifts of the RIS is assumed to be random and the PPO agent is responsible for designing the beamforming vectors. The optimisation problem in this case is as follows
 \begin{subequations} \label{eq:opt4}
 \begin{align}
  \mathcal{P}_4: \hspace{5mm} \max_{\mathbf{f}_u} \hspace{5mm} &\sum_{u=1}^{U}  \log_2 \left(1 + \frac{\frac{P}{U} \mathbf{f}_u^H  \mathbf{C}_u \mathbf{f}_u}{\sigma_2 + \sum_{i \neq u} \frac{P}{U} \mathbf{f}_i^H \mathbf{C}_u \mathbf{f}_i} \right) \tag{\ref{eq:opt4}} \\
 \text{subject to} \hspace{10mm}  &\Vert \mathbf{f}_u \Vert_2^2 = 1, \hspace{3mm}u = 1, \dots, U \\
&\boldsymbol{\Theta} \sim \mathcal{U}[0, 2\pi).
 \end{align}
 \end{subequations}
In order to solve (\ref{eq:opt4}) using PPO, at the beginning of each training episode, the phase shifts of the RIS initialised randomly and keep fixed during the episode and the agent is responsible for optimising the beamforming vectors. In this case, the observation space consist of real and imaginary part of the direct channel and the channel between BS and RIS and that from RIS and UE. The second part of the observation space is the real and imaginary parts of beamfoming vectors. Finally the third part of the observation space is the sum rate of the users. The observation shape is thus $4NU + 2NM + 2UM + 1$.

Next, the action at each time step $t$ is the real and imaginary parts of the beamforming vectors. Hence, the action space is $2NU$.
\begin{figure}[t]
    \centering
    \includegraphics[scale=0.75]{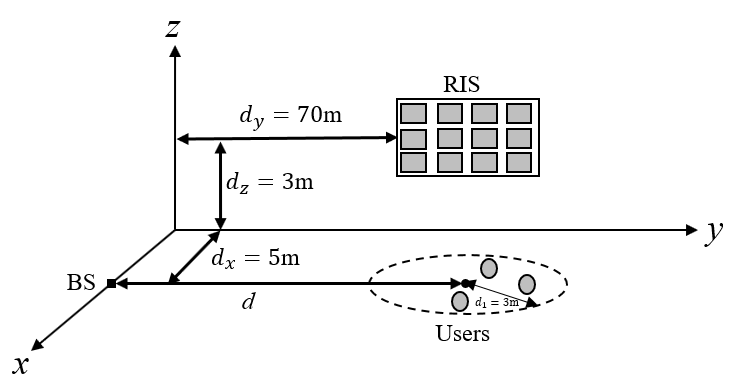}
    \caption{Environment setup}
    \label{env}
\end{figure}
Finally, so as before, the reward of the agent is the sum rate of the users.  
    \item \textbf{Baseline 3} (A2C-based optimisation): For evaluating the effectiveness of the PPO, the environment is trained with another deep reinforcement learning approach named A2C \cite{mnih2016asynchronous}. In this scenario, the observation space, action space and the reward function are the same as in S-CSI-based PPO.
     \item \textbf{Baseline 4} (I-CSI based optimisation): In order to evaluate the performance gap between the case when S-CSI is available at the BS and that when I-CSI is available at the BS, we used the PPO to joint active ans passive beamforming when the I-CSI is fully available at the BS. In this case the optimisation problem is as follows
       \begin{subequations} \label{eq:opt41}
 \begin{align}
   \mathcal{P}_5: \hspace{5mm} \max_{\boldsymbol{\Theta}, \mathbf{f}_u} \hspace{5mm} &\sum_{u=1}^{U}  {\log_2 \left(1 + \frac{\frac{P}{U} \mathbf{f}_u^H \mathbf{h}_u^* \mathbf{h}_u^T \mathbf{f}_u}{\sigma_2 + \sum_{i \neq u} \frac{P}{U} \mathbf{f}_i^H \mathbf{h}_u^* \mathbf{h}_u^T \mathbf{f}_i} \right)} \tag{\ref{eq:opt41}} \\
 \text{subject to} \hspace{10mm} &|\xi_m|^2 = 1, \hspace{3mm}m = 1, \dots, M \\
 &\Vert \mathbf{f}_u \Vert_2^2 = 1, \hspace{3mm}u = 1, \dots, U
 \end{align}
 \end{subequations}
In this case, it is assumed that the I-CSI is initialised at the start of each episode and is kept fixed within each episode.
\end{itemize}

\subsubsection{Convergence Behaviours}
Fig.~\ref{rew} shows the convergence of the PPO-based algorithm as a function of number of episodes. 
The reward curve is obtained by the cumulative rewards obtain from each episode, i.e.,
\begin{equation}
    r_t^{\prime} = \sum_{t=1}^{T^{\prime}} r_t, 
\end{equation}
also, the smoothed learning curve also obtained with the following formula
\begin{equation}
    r_t^{\mathrm{s}} = \frac{1}{k} \sum_{i = t-k+1}^{t} r_i
\end{equation}
where $T^{\prime}$ is the total time steps in each episode and $k$ is the moving average window length. Moreover, the smoothed curve obtained with running average over previous $k = 100$ steps. As it is obvious from Fig.~\ref{rew}, about $500$ episodes in sufficient for the PPO to be converged. Furthermore, the PPO-based approach outperformed the A2C-based method. Also, as the advantage of the PPO is about its robustness against taking many random actions which is the result of of using the $\mathrm{clip}$ function, it is obvious from Fig.~\ref{rew} that the fluctuation of PPO-based curve is significantly less the A2C-based approach.  
\begin{figure}[t]
    \centering
    \includegraphics[scale=0.75]{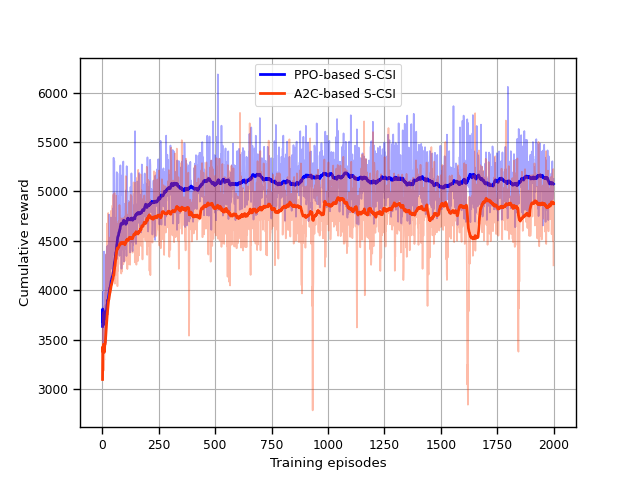}
    \caption{Illustration of the reward as a function of training episodes}
    \label{rew}
\end{figure}
The running time of the PPO-based approach is listed in Table~\ref{table:run}. It is shown that the learning time for S-CSI based PPO is significantly lower than that of I-CSI based. 
\begin{table*}[th]
\caption{The running time} 
\centering 
\begin{tabular}{l l l} 
\hline\hline 
Algorithm & Train time & Test time \\ [0.5ex] 
\hline 
PPO-based I-CSI & $8.2$ hours on CPU & $7.35$ seconds \\
PPO-based S-CSI (Multi user)  & $4.5$ hours on CPU & $4.45$ seconds  \\
PPO-based S-CSI (Single user)  & $1.45$ hours on CPU & $1.25$ seconds  \\
PPO-based Random phase shift (Multi user) & $4.25$ hours on CPU & $3.14$ seconds \\
PPO-based Random phase shift (Single user) & $1.12$ hours on CPU & $1.32$ seconds \\
PPO-based without RIS (Multi user) & $2.2$ hours on CPU & $1.63$ seconds \\
PPO-based without RIS (Single user) & $0.48$ hours on CPU & $0.52$ seconds \\
 [1ex]
\hline 
\end{tabular}
\label{table:run} 
\end{table*}

\subsubsection{Impact of RIS-UE distance}
\begin{figure}[t]
    \centering
    \includegraphics[scale=0.75]{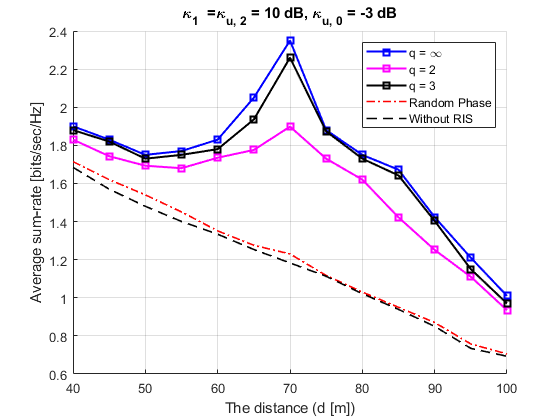}
    \caption{Average sum-rate versus the RIS-UE distance ($d$)}
    \label{dist}
\end{figure}
In practice, the phase shifts of the RIS are quantized, so the passive beamforming is done with discrete values rather than continuous values for the phase shifters. In particular, if the RIS is quantized with $q$ bits, the values of the phase shifters could just take the values of the set $\mathcal{Q} = \{0, \frac{2 \pi}{2^q}, \dots, \frac{2 \pi (2^q-1)}{2^q} \}$, where $q$ is the number of quantization bits. In Fig.~\ref{dist}, the effect of the RIS-UE distance on the average sum rate of the users is shown. For finding the discrete values of the phase shifts of the RIS, we set $\hat{\xi}_i = g(\xi_m)$, where the function $g(\xi_m)$ maps the continuous phase shifts of the RIS, $\xi_m$, to its nearest point in $\mathcal{Q}$, that is
\begin{equation}
    g(\xi_m) = \hat{\xi}_i, \hspace{1mm} \text{if} \hspace{1mm} |\xi_m - \hat{\xi}_i| \leq |\xi_m - \hat{\xi}_j|, \forall \hat{\xi}_i, \hat{\xi}_j \in \mathcal{Q}, \forall i\neq j.
\end{equation}
For training the agent, at the environment setup of Fig.~\ref{env}, the number of users reduced to one user and the single user is located at the centre of the cell. Then the agent trained with different values of $d$ ranging from $d = 40 \mathrm{m}$ to $d = 100 \mathrm{m}$ with steps of $5\mathrm{m}$. The BS maximum power budget is set to be $P = 5\mathrm{dBm}$ and $\kappa_1 = \kappa_{u, 2} = 10 \mathrm{dB}$ and $\kappa_{u, 0} = -3 \mathrm{dB}$ where obviously $u = 1$. Fig.~\ref{kappa} shows how the performance of the RIS is significantly higher when the user is within close proximity of the RIS, even with 2-bit quantized RIS; The 2-bit quantized RIS, on the other hand, performs significantly worse than the 3-bit quantized and the ideal case, however, this gap can be fulfilled with 3-bits quantized RIS which has a close performance in comparison with the ideal case with continuous phase shifters.

\subsubsection{Impact of the transmitted power by the BS}
\begin{figure}[t]
    \centering
    \includegraphics[scale=0.75]{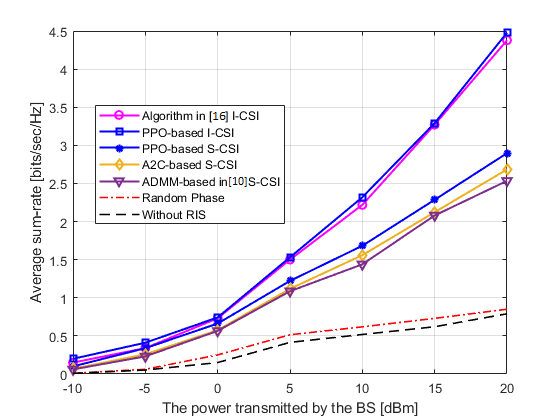}
    \caption{The performance comparison as a function of BS transmit power}
    \label{power}
\end{figure}
In this part, we analyse the impact of BS transmit power. In this scenario, the agent is trained separately with different transmit powers ranging from $-10 \mathrm{dBm}$ to $20\mathrm{dBm}$ with steps of $5\mathrm{dBm}$. Furthermore, the Rician factors are set to be  $\kappa_{1} = \kappa_{u, 2} = \kappa$ and $\kappa_{u, 0} = 0$ and the simulation is done with $3$ number of users. We also compare the performance of the PPO-base approach with the ADMM-based approach proposed in \cite{gan2021ris}. Also, for the validation of I-CSI-based method, the performance is compared with the algorithm proposed in \cite{guo2020weighted}. 

As it is obvious from Fig.~\ref{power}, the performance of all algorithms is the same at low SNR; however, as SNR increases, the I-CSI-based algorithm outperforms the others, while the PPO-based S-CSI algorithm outperforms the A2C-based and ADMM-based algorithms. In addition, all algorithms perform significantly better than random phase shifts of RIS and the case in the absence of RIS.

\subsubsection{Impact of the Rician Factor}

\begin{figure}[t]
    \centering
    \includegraphics[scale=0.75]{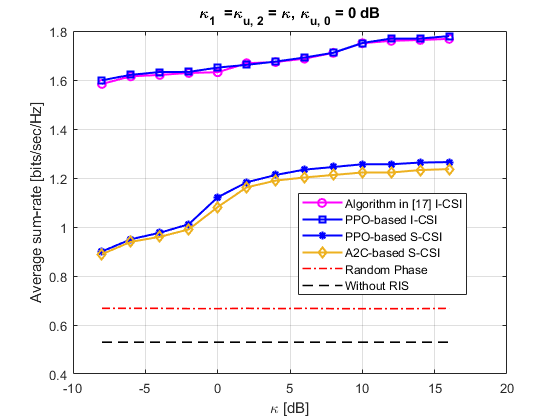}
    \caption{Average sum-rate versus the Rician factor}
    \label{kappa}
\end{figure}
In Fig.~\ref{kappa}, the effect of the Rician factor on the average sum-rate is investigated. To better reflect the impact of Rician factor, it is assumed that $\kappa_{1} = \kappa_{u, 2} = \kappa$ and $\kappa_{u, 0} = 0$. For the I-CSI scheme, we also provide the algorithm proposed in \cite{guo2020weighted}. It is revealed from Fig.~\ref{kappa} that the performance of all the algorithms with S-CSI and I-CSI improves when the Ricisn factor increases. For the S-CSI approach, in particular, this is because as $\kappa$ increases, the BS-RIS-users link becomes more deterministic, which means the LoS link becomes more dominant. It is also observed in the simulation that the gap between I-CSI and S-CSI cases eventually reaches a constant. This is because in the simulation, the direct link between the BS and users is assumed to be fully Rayleigh, and therefore no statistical information can be extracted to further improve performance in high Racian factors. Additionally, this gap will not approach zero due to multiuser interference when no I-CSI can be used for designing the active and passive beamformers. It is also observed that the performance of the algorithm proposed in \cite{guo2020weighted} is similar to PPO-based approach when the I-CSI is available at the BS. Finally, in the random phase-shift case and in the case without RIS, the average sum rate is insensitive to Rician factor.

\section{Conclusion}
In this paper, we have proposed a PPO-based algorithm for joint active and passive beamforming for RIS-aided multiuser MISO systems. A statistical CSI is used here to design both the beamforming vectors at the BS and the phase shifts at the RIS. Furthermore, based on the simulation results, in the low and moderate SNR regimes, statistical CSI-based models achieve comparable performance to instantaneous CSI-based models. In addition, simulation results show that the proposed algorithm is able to converge quickly.

 \section*{Appendix A}
\section*{Proof of Theorem I}
To calculate the $\mathbb{E}[ \mathbf{h}_u^* \mathbf{h}_u^T ]$ we begin with finding $\mathbf{h}_u^* \mathbf{h}_u^T$ where $\mathbf{h}_u^T \triangleq  \mathbf{h}_{u,0}^T + \mathbf{h}_{u,2}^T \boldsymbol{\Xi} \mathbf{H}_1 $
\begin{align}  \nonumber
   \mathbf{h}_u^* \mathbf{h}_u^T  &= \left( \mathbf{h}_{u,0}^* + \mathbf{H}_1^H \boldsymbol{\Xi}^* \mathbf{h}_{u,2}^* \right)  \left( \mathbf{h}_{u,0}^T + \mathbf{h}_{u,2}^T \boldsymbol{\Xi} \mathbf{H}_1 \right) \\ \nonumber
   &= \underbrace{\mathbf{h}_u^* \mathbf{h}_u^T}_{\mathbf{A}} + \underbrace{\mathbf{h}_u^* \mathbf{h}_{u,2}^T \boldsymbol{\Xi} \mathbf{H}_1}_{\mathbf{B}} +  \underbrace{\mathbf{H}_1^H \boldsymbol{\Xi}^* \mathbf{h}_{u,2}^* \mathbf{h}_{u,0}^T}_{\mathbf{C}} + \underbrace{\mathbf{H}_1^H \boldsymbol{\Xi}^* \mathbf{h}_{u,2}^* \mathbf{h}_{u,2}^T \boldsymbol{\Xi} \mathbf{H}_1}_{\mathbf{D}}
\end{align}
By considering $\mathbf{\tilde{h}}_{u,2}$, $ \Tilde{\mathbf{h}}_{u,0}$ and $ \Tilde{\mathbf{H}}_1$ are independently distributed complex Gaussian random matrices with zero mean and unit variance, we will have
\begin{align}
    \mathbb{E}[\mathbf{\tilde{h}}_{u,2}] = \mathbf{0}_{M \times 1} \\
    \mathbb{E}[\Tilde{\mathbf{h}}_{u,0}] = \mathbf{0}_{N \times 1} \\
    \mathbb{E}[\Tilde{\mathbf{H}}_1] = \mathbf{0}_{M \times N}
\end{align}
Next, we begin with calculation of $\mathbb{E}[\mathbf{A}]$ as follows
\begin{align} \label{one}
    \mathbb{E}[\mathbf{A}] &= \mathbb{E} \left[\frac{ \delta_{u,0} \kappa_{u,0}}{1+\kappa_{u,0}} \Bar{\mathbf{h}}_{u,0}^* \Bar{\mathbf{h}}_{u,0}^T \right] + \mathbb{E}\left[\sqrt{\frac{ \delta_{u,0} \kappa_{u,0}}{1+\kappa_{u,0}}}\sqrt{\frac{\delta_{u,0}}{1+\kappa_{u,0}}} \Bar{\mathbf{h}}_{u,0}^* \Tilde{\mathbf{h}}_{u,0}^T\right] \\ \nonumber
    &+ \mathbb{E}\left[\sqrt{\frac{ \delta_{u,0} \kappa_{u,0}}{1+\kappa_{u,0}}}\sqrt{\frac{\delta_{u,0}}{1+\kappa_{u,0}}} \Tilde{\mathbf{h}}_{u,0}^* \Bar{\mathbf{h}}_{u,0}^T\right] + \mathbb{E}\left[\frac{\delta_{u,0}}{1+\kappa_{u,0}} \Tilde{\mathbf{h}}_{u,0}^* \Tilde{\mathbf{h}}_{u,0}^T\right],
\end{align}
obviously, the second and the third term of (\ref{one}) is zero, on the other hand the first term is a constant and the last term is the definition of the covariance matrix, thus
\begin{equation}
    \mathbb{E}[\mathbf{A}] = \frac{ \delta_{u,0} \kappa_{u,0}}{1+\kappa_{u,0}} \Bar{\mathbf{h}}_{u,0}^* \Bar{\mathbf{h}}_{u,0}^T + \frac{\delta_{u,0}}{1+\kappa_{u,0}} \mathbf{I}_N
\end{equation}
By following the same approach and noting that for a Gaussian distributed zero mean and unit variance matrix $\mathbf{X} \in \mathbb{C}^{P \times Q}$ $\mathbf{Z}$, $\mathbb{E}[\mathbf{X}^H \mathbf{Z} \mathbf{X}] = \mathrm{trace}(\mathbf{Z}) \mathbf{I}_Q$ and with some simple calculations the remaining terms could be calculated and listed as follows

\begin{align}
    \mathbb{E}[\mathbf{B}] &=  \sqrt{\frac{\delta_1 \delta_{u,0} \delta_{u,2} \kappa_1 \kappa_{u,0} \kappa_{u, 2}}{( 1+\kappa_1 )(1+\kappa_{u, 0})(1+\kappa_{u,2})}}  \Bar{\mathbf{h}}_{u,0}^* \Bar{\mathbf{h}}_{u,2}^T \boldsymbol{\Xi} \Bar{\mathbf{H}}_1 \\
    \mathbb{E}[\mathbf{C}] &=  \sqrt{\frac{\delta_1 \delta_{u,0} \delta_{u,2} \kappa_1 \kappa_{u,0} \kappa_{u, 2}}{( 1+\kappa_1 )(1+\kappa_{u, 0})(1+\kappa_{u,2})}}\Bar{\mathbf{H}}_1^H \boldsymbol{\Xi}^* \Bar{\mathbf{h}}_{u,2}^* \Bar{\mathbf{h}}_{u,0}^T \\
    \mathbb{E}[\mathbf{D}] &=  \frac{M\delta_{1}\delta_{u,2}}{1 + \kappa_1}\mathbf{I}_N + \frac{M\delta_{u,2} \delta_{1} \kappa_{1}}{(1 + \kappa_{1})(1 + \kappa_{u, 2})}\mathbf{a}_\text{BS}(\phi^{(\text{BS})}, \psi^{(\text{BS})})^H \mathbf{a}_\text{BS}(\phi^{(\text{BS})}, \psi^{(\text{BS})}) \\ \nonumber  
    &+ \frac{\delta_{1} \delta_{u,2}  \kappa_{1} \kappa_{u,2}}{(1 +  \kappa_{1})(1 +  \kappa_{u, 2})} \Bar{\mathbf{H}}_1^H \boldsymbol{\Xi}^* \Bar{\mathbf{h}}_{u,2}^* \Bar{\mathbf{h}}_{u,2}^T \boldsymbol{\Xi}  \Bar{\mathbf{H}}_1
\end{align}
Finally, by summing up all the terms, the result will be obtained which completes the proof.
\ifCLASSOPTIONcaptionsoff
  \newpage
\fi

\bibliographystyle{IEEEtran}
\bibliography{refs}

\begin{thebibliography}{10}
\providecommand{\url}[1]{#1}
\csname url@samestyle\endcsname
\providecommand{\newblock}{\relax}
\providecommand{\bibinfo}[2]{#2}
\providecommand{\BIBentrySTDinterwordspacing}{\spaceskip=0pt\relax}
\providecommand{\BIBentryALTinterwordstretchfactor}{4}
\providecommand{\BIBentryALTinterwordspacing}{\spaceskip=\fontdimen2\font plus
\BIBentryALTinterwordstretchfactor\fontdimen3\font minus
  \fontdimen4\font\relax}
\providecommand{\BIBforeignlanguage}[2]{{%
\expandafter\ifx\csname l@#1\endcsname\relax
\typeout{** WARNING: IEEEtran.bst: No hyphenation pattern has been}%
\typeout{** loaded for the language `#1'. Using the pattern for}%
\typeout{** the default language instead.}%
\else
\language=\csname l@#1\endcsname
\fi
#2}}
\providecommand{\BIBdecl}{\relax}
\BIBdecl

\bibitem{wu2019towards}
Q.~Wu and R.~Zhang, ``Towards smart and reconfigurable environment: Intelligent
  reflecting surface aided wireless network,'' \emph{IEEE Communications
  Magazine}, vol.~58, no.~1, pp. 106--112, 2019.

\bibitem{huang2019reconfigurable}
C.~Huang, A.~Zappone, G.~C. Alexandropoulos, M.~Debbah, and C.~Yuen,
  ``Reconfigurable intelligent surfaces for energy efficiency in wireless
  communication,'' \emph{IEEE Transactions on Wireless Communications},
  vol.~18, no.~8, pp. 4157--4170, 2019.

\bibitem{di2020smart}
M.~Di~Renzo, A.~Zappone, M.~Debbah, M.-S. Alouini, C.~Yuen, J.~De~Rosny, and
  S.~Tretyakov, ``Smart radio environments empowered by reconfigurable
  intelligent surfaces: How it works, state of research, and the road ahead,''
  \emph{IEEE Journal on Selected Areas in Communications}, vol.~38, no.~11, pp.
  2450--2525, 2020.

\bibitem{wu2019intelligent}
Q.~Wu and R.~Zhang, ``Intelligent reflecting surface enhanced wireless network
  via joint active and passive beamforming,'' \emph{IEEE Transactions on
  Wireless Communications}, vol.~18, no.~11, pp. 5394--5409, 2019.

\bibitem{wu2019beamforming}
------, ``Beamforming optimization for wireless network aided by intelligent
  reflecting surface with discrete phase shifts,'' \emph{IEEE Transactions on
  Communications}, vol.~68, no.~3, pp. 1838--1851, 2019.

\bibitem{wu2018intelligent}
------, ``Intelligent reflecting surface enhanced wireless network: Joint
  active and passive beamforming design,'' in \emph{2018 IEEE Global
  Communications Conference (GLOBECOM)}.\hskip 1em plus 0.5em minus 0.4em\relax
  IEEE, 2018, pp. 1--6.

\bibitem{song2020unsupervised}
H.~Song, M.~Zhang, J.~Gao, and C.~Zhong, ``Unsupervised learning-based joint
  active and passive beamforming design for reconfigurable intelligent surfaces
  aided wireless networks,'' \emph{ieee communications letters}, vol.~25,
  no.~3, pp. 892--896, 2020.

\bibitem{huang2020reconfigurable}
C.~Huang, R.~Mo, and C.~Yuen, ``Reconfigurable intelligent surface assisted
  multiuser miso systems exploiting deep reinforcement learning,'' \emph{IEEE
  Journal on Selected Areas in Communications}, vol.~38, no.~8, pp. 1839--1850,
  2020.

\bibitem{marzetta2010noncooperative}
T.~L. Marzetta, ``Noncooperative cellular wireless with unlimited numbers of
  base station antennas,'' \emph{IEEE transactions on wireless communications},
  vol.~9, no.~11, pp. 3590--3600, 2010.

\bibitem{rusek2012scaling}
F.~Rusek, D.~Persson, B.~K. Lau, E.~G. Larsson, T.~L. Marzetta, O.~Edfors, and
  F.~Tufvesson, ``Scaling up mimo: Opportunities and challenges with very large
  arrays,'' \emph{IEEE signal processing magazine}, vol.~30, no.~1, pp. 40--60,
  2012.

\bibitem{jose2011pilot}
J.~Jose, A.~Ashikhmin, T.~L. Marzetta, and S.~Vishwanath, ``Pilot contamination
  and precoding in multi-cell tdd systems,'' \emph{IEEE Transactions on
  Wireless Communications}, vol.~10, no.~8, pp. 2640--2651, 2011.

\bibitem{wang2015acquisition}
G.~Wang, Q.~Liu, R.~He, F.~Gao, and C.~Tellambura, ``Acquisition of channel
  state information in heterogeneous cloud radio access networks: Challenges
  and research directions,'' \emph{IEEE Wireless Communications}, vol.~22,
  no.~3, pp. 100--107, 2015.

\bibitem{adhikary2013joint}
A.~Adhikary, J.~Nam, J.-Y. Ahn, and G.~Caire, ``Joint spatial division and
  multiplexing—the large-scale array regime,'' \emph{IEEE transactions on
  information theory}, vol.~59, no.~10, pp. 6441--6463, 2013.

\bibitem{dang2020joint}
J.~Dang, Z.~Zhang, and L.~Wu, ``Joint beamforming for intelligent reflecting
  surface aided wireless communication using statistical csi,'' \emph{China
  Communications}, vol.~17, no.~8, pp. 147--157, 2020.

\bibitem{you2020channel}
C.~You, B.~Zheng, and R.~Zhang, ``Channel estimation and passive beamforming
  for intelligent reflecting surface: Discrete phase shift and progressive
  refinement,'' \emph{IEEE Journal on Selected Areas in Communications},
  vol.~38, no.~11, pp. 2604--2620, 2020.

\bibitem{wang2020channel}
Z.~Wang, L.~Liu, and S.~Cui, ``Channel estimation for intelligent reflecting
  surface assisted multiuser communications,'' in \emph{2020 IEEE Wireless
  Communications and Networking Conference (WCNC)}.\hskip 1em plus 0.5em minus
  0.4em\relax IEEE, 2020, pp. 1--6.

\bibitem{ning2020channel}
B.~Ning, Z.~Chen, W.~Chen, and Y.~Du, ``Channel estimation and transmission for
  intelligent reflecting surface assisted thz communications,'' in \emph{ICC
  2020-2020 IEEE International Conference on Communications (ICC)}.\hskip 1em
  plus 0.5em minus 0.4em\relax IEEE, 2020, pp. 1--7.

\bibitem{gan2021ris}
X.~Gan, C.~Zhong, C.~Huang, and Z.~Zhang, ``Ris-assisted multi-user miso
  communications exploiting statistical csi,'' \emph{IEEE Transactions on
  Communications}, vol.~69, no.~10, pp. 6781--6792, 2021.

\bibitem{han2019large}
Y.~Han, W.~Tang, S.~Jin, C.-K. Wen, and X.~Ma, ``Large intelligent
  surface-assisted wireless communication exploiting statistical csi,''
  \emph{IEEE Transactions on Vehicular Technology}, vol.~68, no.~8, pp.
  8238--8242, 2019.

\bibitem{hu2020statistical}
X.~Hu, J.~Wang, and C.~Zhong, ``Statistical csi based design for intelligent
  reflecting surface assisted miso systems,'' \emph{Science China Information
  Sciences}, vol.~63, no.~12, pp. 1--10, 2020.

\bibitem{zhao2021two}
M.-M. Zhao, A.~Liu, Y.~Wan, and R.~Zhang, ``Two-timescale beamforming
  optimization for intelligent reflecting surface aided multiuser communication
  with qos constraints,'' \emph{IEEE Transactions on Wireless Communications},
  vol.~20, no.~9, pp. 6179--6194, 2021.

\bibitem{schulman2017proximal}
J.~Schulman, F.~Wolski, P.~Dhariwal, A.~Radford, and O.~Klimov, ``Proximal
  policy optimization algorithms,'' \emph{arXiv preprint arXiv:1707.06347},
  2017.

\bibitem{wang2021joint}
J.~Wang, H.~Wang, Y.~Han, S.~Jin, and X.~Li, ``Joint transmit beamforming and
  phase shift design for reconfigurable intelligent surface assisted mimo
  systems,'' \emph{IEEE Transactions on Cognitive Communications and
  Networking}, vol.~7, no.~2, pp. 354--368, 2021.

\bibitem{zhang2020capacity}
S.~Zhang and R.~Zhang, ``Capacity characterization for intelligent reflecting
  surface aided mimo communication,'' \emph{IEEE Journal on Selected Areas in
  Communications}, vol.~38, no.~8, pp. 1823--1838, 2020.

\bibitem{zhou2020joint}
Z.~Zhou, N.~Ge, Z.~Wang, and L.~Hanzo, ``Joint transmit precoding and
  reconfigurable intelligent surface phase adjustment: A decomposition-aided
  channel estimation approach,'' \emph{IEEE Transactions on Communications},
  vol.~69, no.~2, pp. 1228--1243, 2020.

\bibitem{liu2018spectral}
P.~Liu, K.~Luo, D.~Chen, T.~Jiang, and M.~Matthaiou, ``Spectral efficiency
  analysis of multi-cell massive mimo systems with ricean fading,'' in
  \emph{2018 10th International Conference on Wireless Communications and
  Signal Processing (WCSP)}.\hskip 1em plus 0.5em minus 0.4em\relax IEEE, 2018,
  pp. 1--7.

\bibitem{mnih2016asynchronous}
V.~Mnih, A.~P. Badia, M.~Mirza, A.~Graves, T.~Lillicrap, T.~Harley, D.~Silver,
  and K.~Kavukcuoglu, ``Asynchronous methods for deep reinforcement learning,''
  in \emph{International conference on machine learning}.\hskip 1em plus 0.5em
  minus 0.4em\relax PMLR, 2016, pp. 1928--1937.

\bibitem{guo2020weighted}
H.~Guo, Y.-C. Liang, J.~Chen, and E.~G. Larsson, ``Weighted sum-rate
  maximization for reconfigurable intelligent surface aided wireless
  networks,'' \emph{IEEE Transactions on Wireless Communications}, vol.~19,
  no.~5, pp. 3064--3076, 2020.

\end{thebibliography}

\end{document}